%% file: main.tex
\documentclass[acmsmall, nonacm, authorversion, screen]{acmart}
\usepackage{xspace, nicefrac, ctable, enumerate, bbm}
\usepackage{algpseudocode}
\usepackage{algorithm}
\usepackage{subcaption}
\usepackage{thmtools}
\usepackage{nicefrac}
\usepackage{mathtools}

\newcommand{\C}[1]{\ensuremath{\mathrm{Cost}(#1)}}
\newcommand{\rob}{\ensuremath{\texttt{PAWS}}\xspace}
\newcommand{\decision}{\ensuremath{\texttt{PAWS-DA}}\xspace}

\newcommand{\paowt}{\ensuremath{\texttt{PA-OWT}}\xspace}
\newcommand{\owt}{\ensuremath{\texttt{OWT}}\xspace}
\newcommand{\mix}{\ensuremath{\texttt{Mix}}\xspace}

\newcommand{\opt}{\ensuremath{\texttt{OPT}}\xspace}

\newcommand{\alg}{\ensuremath{\texttt{ALG}}\xspace}
\newcommand{\Opt}{\texttt{OPT}\xspace}
\newcommand{\adv}{\texttt{ADV}\xspace}
\newcommand{\Alg}{\texttt{ALG}\xspace}

\algnewcommand{\Initialize}[1]{%
  \State \textbf{initialize:}
  \Statex \hspace*{\algorithmicindent}\parbox[t]{.8\linewidth}{\raggedright #1}
}

\newcommand{\sref}[2]{\hyperref[#2]{#1 \ref{#2}}}

\declaretheorem[name=Assumption]{assumption}

\newcommand{\N}{\mathbb{N}}

\usepackage{graphicx} %

\AtBeginDocument{%
  }

\title[Peak-Aware One-Way Trading]{Peak-Aware One-Way Trading: Competitive and Learning-Augmented Algorithms for Energy Procurement}

\author{Jackson Bibbens}
\affiliation{%
  \institution{University of Massachusetts Amherst}
  \country{USA}
}
\email{jbibbens@umass.edu}

\author{Adam Lechowicz}
\affiliation{%
  \institution{University of Massachusetts Amherst}
  \country{USA}
}
\email{alechowicz@umass.edu}

\author{Nicolas Christianson}
\affiliation{%
  \institution{Johns Hopkins University}
  \country{USA}
}
\email{christianson@jhu.edu}

\author{Bo Sun}
\affiliation{%
  \institution{University of Ottawa}
  \country{Canada}
}
\email{bo.sun@uottawa.ca}

\author{Mohammad Hajiesmaili}
\affiliation{%
  \institution{University of Massachusetts Amherst}
  \country{USA}
}
\email{hajiesmaili@cs.umass.edu}

\begin{abstract}
    \input{sec-abstract}

\end{abstract}

\begin{document}

\maketitle

\section{Introduction}
\label{sec:intro}

\input{1-intro.tex}

\subsection{Related Work}
\label{sec:relwork}

\input{sec-relwork.tex}

\section{Problem and Preliminaries}
\label{sec:prob}

\input{sec-problem.tex}

\section{Competitive Algorithms for \paowt}\label{sec:rob}
\input{sec-robust-alg}

\section{A Fundamental Limit}\label{sec:lower-bound}

\input{sec_preliminaries}

\section{Learning-Augmented Algorithms for \paowt}\label{sec:decision}
\input{sec_learning-aug-algs}

\section{Experiments}\label{sec:experiments}
\input{sec-experiments}

\section{Conclusion}

\input{sec-conclusion}

\begin{acks}
\input{acknowledgments}
\end{acks}

\bibliographystyle{ACM-Reference-Format}
\bibliography{references}

\section*{Artificial Intelligence Use Acknowledgment}
\input{sec-ai-use}

\appendix

\section{Proofs of Preliminaries}\label{app:prelim}
\input{appendix-preliminaries}

\section{Proofs of Learning-Augmented Algorithms}\label{app:LAA}
\input{appendix-learning-aug-algs}

\end{document}

%% file: sec-abstract.tex
We introduce \textit{peak-aware one-way trading} (\paowt), an online problem in which a player must purchase a fixed quantity of an asset by a deadline, minimizing its total cost plus a charge proportional to the largest single purchase at any time.
The model is motivated by electricity procurement, where the cost of energy rewards buying at low prices, but peak demand charges discourage concentrating purchases in short periods.
Balancing these incentives is challenging when future prices are unknown: a large purchase can establish a peak that later decisions cannot undo.
We develop a peak-aware online algorithm that balances price and peak control with deadline feasibility and achieves the order-optimal competitive ratio. 
We further extend our results to a learning-augmented algorithm that uses predictions of optimal purchasing decisions to achieve an order-optimal consistency-robustness tradeoff, exploiting accurate advice while retaining worst-case guarantees in the face of arbitrary prediction errors.
Finally, we evaluate our algorithms in a case study on electricity procurement using historical electricity price traces to understand the benefits of peak-aware purchasing.

%% file: 1-intro.tex
\noindent In the classic one-way trading problem, an online player buys/sells shares of an item over a fixed time horizon.  They observe each price sequentially and must decide how much to trade before seeing the next price.  In the purchasing variant, the player must acquire a fixed quantity before the deadline while minimizing their total cost. This paper introduces \textit{peak-aware one-way trading} (\paowt), which adds a second consideration: the player pays not only for the shares purchased, but also for the largest purchase made in any single round. Specifically, the player purchases at least one unit over a horizon of $T$ rounds and minimizes $\sum_{t=1}^{T} p_t x_t + C \max_{t \in [T]} x_t$,
where $p_t$ is the price, $x_t$ is the (fractional) purchase amount, and $C > 0$ is the peak-charge coefficient. The addition of the peak term creates a tension absent from the classic one-way trading setting: a low price encourages buying more now, but doing so can establish a peak charge that no later decision can undo.

This problem is motivated by electricity procurement for large energy consumers, such as data centers or factories.  Commercial electricity bills can include both an energy charge, based on consumption, and a demand charge, based on the customer's highest power demand during the billing period~\cite{REopt-UtilityRates}. These charges reward different behaviors. Time-varying energy prices encourage an operator to move consumption to cheaper intervals, whereas a demand charge encourages the operator to spread consumption out over time. For example, an enterprise charging a fleet of electric vehicles before the start of the workday can reduce energy costs by concentrating charging in a low-price interval, but doing so will increase the demand component of their energy bill.

The \paowt model is related to two lines of work in the online optimization literature. In online search problems such as one-way trading and $k$-search, there is a long-term fixed demand which can be met flexibly across the time horizon, and the goal is to purchase enough supply to satisfy the demand while minimizing the total purchasing cost in the face of uncertain future prices~\cite{ElYaniv:01,Lorenz:09,Cao:20,SunZeynali:20}. Viewing the problem as one of allocating a fixed trading budget across prices that arrive online, these problems have an online packing structure~\cite{Buchbinder:09, Zhou:08}. To the best of our knowledge, costs associated with peak purchases (i.e., demand charges) have not been studied in this line of work, and significantly complicate algorithm design due to their coupling the cost of decisions across time. In contrast, studies of peak-aware economic dispatch~\cite{Zhang:15,Zhang:18} explicitly account for demand charges. However, in this line of work, demands are uncertain and arrive over time, and the decisions at each time are whether to purchase supply from the grid (which includes both purchasing cost and demand charge) or from an alternative supply, such as energy storage~\cite{Mo:26:Automatica} or expensive local generation~\cite{Zhang:15,Zhang:18,Lee:21Peak, Himmelreich:26}. However, nearly all models in this line of work have an online covering structure~\cite{Azar:16,Nagarajan:17}, where demand must be met exactly at each stage; while some works also consider flexible demand, the proposed algorithms are typically only heuristic, and do not come with formal performance guarantees for the setting with both flexible demand and time-varying prices~\cite{Wu:22EV, Lee:20, Mo:26:Automatica}.

\paowt sits at the intersection of these tracks of literature, considering the question of how to optimally handle both a flexible, horizon-wide purchasing requirement and a peak penalty under time-varying prices.
Buying aggressively at low prices risks an unnecessarily large peak if many cheap prices arrive later, while spreading purchases out more cautiously risks paying higher prices, and postponing too much procurement to the end of the instance can force a large purchase near the deadline. 
Knowledge of the horizon also matters: a player that knows how many rounds remain can plan how to distribute its remaining purchases, whereas a player facing an unknown horizon cannot rely on those opportunities. We therefore ask:
\begin{quote}
    \emph{Is it possible to design competitive algorithms for \paowt that balance a purchasing cost and a peak charge while meeting a deadline?}
\end{quote}

The online algorithms literature typically studies algorithm design using worst-case competitive analysis: comparing an algorithm's total cost with that of an offline optimum that knows the entire input sequence in advance. 
In an electricity market context, worst-case guarantees are useful: prices are known to be non-stationary and challenging to predict~\cite{Qu:24}.
However, algorithms that optimize purely for the worst-case can be overly conservative in practice.  \emph{Learning-augmented algorithms} provide a framework for using untrusted advice about the future (e.g., machine-learned predictions, day-ahead price forecasts, etc.) without losing the worst-case guarantees associated with traditional algorithm design.
Prior related work for one-way trading and peak-aware energy scheduling has shown that such algorithms can use predictions (e.g., of the best price or the optimal decisions at each round) to improve empirical performance without sacrificing worst-case guarantees~\cite{Lee:21Peak, SunLee:21,Johnson:26UQ,Himmelreich:26}.
This naturally motivates our second question:
\begin{quote}
    \emph{Can predictions of the optimal purchasing decisions improve the performance of algorithms for \paowt while preserving a worst-case guarantee?}
\end{quote}

\subsection{Contributions}

We present a comprehensive analysis of the \paowt problem, providing order-optimal online algorithms both for the competitive and learning-augmented settings. In particular, our main contributions are as follows.

\paragraph{Competitive results.} We show that---unlike classic one-way trading (\owt) ---there is an asymptotic separation between the \textit{known-horizon} and \textit{unknown-horizon} settings. When the horizon $T$ is unknown, we show a linear lower bound of $\Omega(T)$ (\autoref{thm:unknown-T-lb}), which can be matched up to constant factors by any algorithm, which motivates a focus on the known-horizon setting for the rest of the paper. In  \autoref{sec:rob}, we present Peak-Aware Water Level Smoothing (\rob), a dual-threshold online algorithm, and show that it achieves the order-optimal  competitive ratio for \paowt: $\sqrt{\nicefrac{U}{(L+C/T)}}$ (Theorems \ref{thm:rob-ub} and \ref{thm:known-T-lb}). This bound is a generalization of the one-way trading competitive ratio $\sqrt{\nicefrac{U}{L}}$ (for the minimization setting) and results in a sublinear dependence on $T$ when $L=0$. Unlike classic one-way trading, \rob only requires fixed thresholds rather than a threshold curve, and can accommodate $L=0$.  We additionally show that \paowt transitions smoothly to \owt as $C\rightarrow 0$ and to trivial peak minimization when $C$ is sufficiently large.

\paragraph{Learning-augmented results.} In \autoref{sec:decision}, we consider the design of learning-augmented algorithms for the \paowt problem. We first give a lower bound on the achievable consistency-robustness trade-off for any learning-augmented algorithm (\autoref{thm:advice-gen-lb}). Specifically, any $(1+\epsilon)$-consistent algorithm requires $\Omega(\max\{\sqrt{\nicefrac{U}{(L+C/T)}},\nicefrac{U}{(1+\epsilon)(L+C/T)}\})$-robustness. We then show that a naive convex combination of \rob and decision advice (\mix) is unable to match this tradeoff, giving robustness linear in $T$ (\autoref{thm:mixAlg}) and necessitating the development of a more sophisticated strategy. Motivated by this gap, we propose an algorithm, \decision, which generates proposed actions based on decision advice and projects those actions onto a robust interval. We show that \decision can achieve any desired consistency and matches the order-optimal trade-off between robustness and consistency (\autoref{thm:decision-advice-consistency}, \autoref{thm:decision-advice-ub}).

\paragraph{Technical insights} 
Algorithms for \owt and online search compare purchases made at a reservation price or threshold curve to an optimal cost equaling the lowest price seen. The addition of a peak charge breaks this analysis, as the offline optimum may now spread purchases across multiple time steps. In \paowt, the offline dual problem elicits a \textit{water filling} solution, where \Opt pays cost at least $\Phi$ whenever the amounts by which prices fall below $\Phi$ sum to at most the peak-charge coefficient $C$. Our competitive algorithm \rob follows this structure; buying proportionally to how far the price falls below a fixed water level, so it either completes its purchase below that level or certifies that the optimum pays at least as much. \rob also introduces a rate limiter to spread compulsory purchases as the deadline nears, ensuring feasibility while preventing a large peak.

Integrating predictions in this framework requires a mechanism to limit the damage of poor advice. The convex combination strategy of \mix --- appropriate in related problems --- here inherits a fixed fraction of the advice's peak regardless of advice quality. By comparison, \decision uses shadow decisions to greedily track advice while the accumulated cost of these decisions (including peak charge) remains within a budget. This reduces the exposure to bad advice from a multiplicative to an additive term. In each time step, \decision will purchase the larger of the shadow and water filling decisions, capped by a rate limiter as in \rob.

\paragraph{Case study.}
In \autoref{sec:experiments}, we evaluate the performance of our proposed algorithms in a case study on electricity procurement in four US grid regions. We observe that day-ahead prices serve as effective predictions, allowing learning-augmented algorithms to perform within $\approx20\%$ of the offline optimal solution on average. Empirical behavior of the algorithms generally follows the theoretical results, despite testing on predictable, non-adversarial instances. Our theoretical results rely on both \textit{price} and \textit{peak charge} separation between algorithms and the offline optimal; we find empirical variation on both dimensions, though slightly more within the peak charge than prices. Similarly, algorithm performance tends to improve as $C$ increases, and degrades as $T$ increases --- both of which align with our theoretical analysis.

%% file: sec-relwork.tex
Our work builds on several related lines of work in the online algorithms literature and application-focused studies.
We first discuss how \paowt directly relates to prior work in online search and peak-aware economic dispatch, before highlighting connections to the broader online optimization literature and relevant application studies.

\paragraph{Online search}
Most directly related to our setting, the classic online problems of one-way trading and $k$-search study sequential purchases (or sales) of a divisible asset at varying prices revealed over time~\cite{ElYaniv:01,Lorenz:09}.  
In the purchasing setting, the player must buy a fixed amount by a deadline, but pays no charge for the largest purchasing decision.
These problems (and related problems such as online knapsack~\cite{Zhou:08}) admit optimal threshold-based algorithms that make decisions based on price observations and remaining budget or capacity~\cite{Cao:20,SunZeynali:20}.
Extensions of these problems have considered switching costs that penalize the player for changing their purchasing decisions over time~\cite{Lechowicz:23PR,Lechowicz:24OCS}, inventory models~\cite{Lin:19,Yang:19}, and nonlinear price functions~\cite{Lechowicz:24LongTerm}.
Unlike these models, \paowt must balance buying at low prices against keeping the largest purchase small.

\paragraph{Peak-aware economic dispatch.}
Several papers study peak charges through an online algorithms lens.
Canonical peak-aware economic dispatch assumes inelastic demand revealed online and chooses between grid purchases and local generation~\cite{Zhang:15,Zhang:18}.
Related work also considers flexible demand. \citet{Wang:13Workload} allow data-center workloads to be delayed or dropped under time-varying prices and peak charges, although they give guarantees only for a case with time-invariant prices where demand is dropped, not delayed. \citet{Zhao:15PeakEV} derive competitive guarantees for peak-minimizing online EV charging with flexible deadlines, but no time-varying prices---their model extends prior work on processor speed scaling~\cite{Bansal:07}.
Other extensions consider inventory~\cite{Johnson:11} and time-varying prices under an assumption that the peak charge exceeds $T \times$ the price range~\cite{Mo:21PeakDemand}---in this regime, focusing exclusively on peak minimization is optimal (see \sref{Lemma}{lem:largeC}).

Compared with these models, \paowt captures a distinct combination of features.  It adopts the maximal temporal flexibility of one-way trading: a fixed divisible demand requirement is known at the outset and may be allocated arbitrarily across the horizon as time-varying prices are revealed.
By allowing demand to be shifted freely across the horizon, \paowt isolates the fundamental online trade-off between \emph{when and how much to purchase} to exploit low prices versus limiting the peak created by those purchases.

\paragraph{Online convex optimization and covering}
\paowt's peak charge couples all purchase decisions through an $\ell_\infty$ norm.
This connects \paowt to online covering with convex and norm objectives, although in these problems the objective is known (i.e., prices are known in advance) and constraint rows arrive online~\cite{Azar:16,Nagarajan:17}.
Other problems such as convex body chasing and smoothed online convex optimization also include norms in their objectives to model the cost of switching between consecutive decisions, but do not consider costs (like a peak charge) that couple decisions across all time steps~\cite{Andrew:13, Sellke:20,Chen:18}.

\paragraph{Learning-augmented algorithms}
Learning-augmented algorithms use predictions to improve performance (\textit{consistency}) while retaining worst-case guarantees (\textit{robustness}).
For one-way trading and related online search models, studies consider \textit{instance advice} (e.g., a single predicted ``best price'' for each instance)~\cite{SunLee:21,Lee:24}. Follow-up work has shown that considering consistency-robustness alone can hide ``brittle'' behavior under small errors in instance advice~\cite{Elenter:24}.
In peak-aware energy scheduling, studies have also considered instance advice~\cite{Lee:21Peak}, with follow-up work that considers risk-sensitivity~\cite{Himmelreich:26}.
In the broader literature on online convex optimization, convex body chasing, and online search problems with switching costs, learning-augmentation has been studied using \textit{decision advice} that predicts the optimal decision in each time step~\cite{pmlr-v119-antoniadis20a,Christianson:22,Lechowicz:24OCS}.  
In our work, we primarily focus on decision advice; we briefly discuss how instance advice might inform the design of learning-augmented algorithms for \paowt in \autoref{sec:decision}.

\paragraph{Applications to electric load scheduling}

By providing worst-case guarantees under arbitrary price sequences, our competitive-analysis framework complements application-focused approaches that rely on forecasts, stochastic models, or control assumptions.
Flexible loads under demand charges have been studied using model predictive control, particularly for EV charging~\cite{Yang:20,Lee:20,Cambronne:26}, as well as through data-center workload shifting and partial execution~\cite{Wang:13Workload,Xu:14,Dabbagh:19}.
\citet{Wu:22EV} studied flexible EV charging under time-varying prices and peak-charge, formulating it as a bi-objective problem that first minimizes the peak, then minimizes energy cost---they used predictions and rolling-horizon optimization to solve this.
Related work has also considered the offline optimization of production scheduling under time-varying energy and demand charges~\cite{BatistaAbikarram:19} and reinforcement learning for demand-charge minimization with storage~\cite{Weber:24}.
These works study the same basic optimization tradeoff between energy cost and peak charges using forecast- or model-driven approaches; our competitive analysis complements them by providing worst-case guarantees relative to the offline optimum without assumptions on the accuracy of future price forecasts.

%% file: sec-problem.tex
In this section, we formalize peak-aware one-way trading, state our assumptions, and define the competitive ratio used throughout the paper.

\subsection{Problem Formulation}

We introduce the \textbf{p}eak-\textbf{a}ware \textbf{o}ne-\textbf{w}ay \textbf{t}rading problem (\paowt), where the player's goal is to purchase a fixed quantity of an asset by the end of a time horizon $T$ while minimizing their purchase cost and a peak charge.
Without loss of generality, the purchase quantity is normalized to one unit.    

At each time step $t\in[T]$, the player observes a unit price $p_t$ and chooses a purchase amount $x_t\in[0,1]$ before observing any future prices.
Each decision is irrevocable: a purchase cannot be canceled or moved to another time step.
The player must purchase at least one unit by the end of the final time step $T$, i.e., $\sum_{t=1}^T x_t\ge1$.

A purchase of $x_t$ units costs $p_t x_t$.
In addition to the total purchase cost, the player pays a single peak charge $C\max_{t\in[T]}x_t$, proportional to its largest purchase, where the peak-charge coefficient $C>0$ is constant and known in advance.
Thus, buying more at a low price may reduce purchase costs but increase the peak charge.
For a price sequence $\mathcal{I}=(p_1,\ldots,p_T)$, the offline problem is
\begin{align}
 \label{eq:paowt-model} 
\min_{\mathbf{x}\in[0,1]^T}\quad \sum_{t\in[T]} p_t x_t + C\max_{t\in[T]}x_t \quad\quad \text{s.t.}\quad \sum_{t\in[T]} x_t\ge1,
\end{align}
where the first term and second term in the objective are the purchase cost and the peak charge, respectively.
In the online setting, the player must meet the same purchase constraint without knowing future prices.
Any amount still needed in the final round must be purchased at $p_T$; this \textit{compulsory trade} contributes to both purchase cost and the peak charge.

\subsection{Preliminaries \& Assumptions}\label{sec:prelims}

We study \paowt within the online algorithms framework.  In this setting, the objective is to design an algorithm that achieves a small \textit{competitive ratio}~\cite{Manasse:88}, defined as follows:

Let $\Alg(\mathcal{I})$ denote the purchase schedule produced by a deterministic online algorithm, and let $\Opt(\mathcal{I})$ denote an optimal solution to \eqref{eq:paowt-model} with full knowledge of the price sequence.
We write their total costs as $\C{\Alg(\mathcal{I})}$ and $\C{\Opt(\mathcal{I})}$, respectively.
Let
$\Omega_T$ be the set of possible \paowt price instances.
The competitive ratio of \Alg on sequences of length $T$ is 

\begin{equation}
    \operatorname{CR}_T(\Alg)
    = \sup_{\mathcal{I}\in\Omega_T}
      \frac{\C{\Alg(\mathcal{I})}}{\C{\Opt(\mathcal{I})}}.
    \label{eq:competitive-ratio}
\end{equation}
We call \Alg $\alpha$-competitive if $\C{\Alg(\mathcal{I})}\le\alpha\,\C{\Opt(\mathcal{I})}$ for every $\mathcal{I}\in\Omega_T$.
Note that $\alpha$, by definition, is at least one, and a smaller ratio gives a stronger worst-case guarantee. We restrict our focus in this work to deterministic algorithms, since convexity of the objective implies that any randomized algorithm can be derandomized (i.e., finding a corresponding deterministic algorithm) without increasing its total cost. This can be seen by following the proof of Theorem~1 in~\cite{ElYaniv:01} and applying the convexity of the objective.

Following the related literature~\cite{ElYaniv:01, SunLee:21, Lorenz:09}, we assume:
\begin{assumption} \label{asm:priceBounds}
     All prices satisfy $p_t\in[L,U]$, where $0\le L\le U<\infty$ are known bounds.
\end{assumption}
Without this assumption, it is known that the competitive ratio of any online algorithm can be arbitrarily large, even in the setting without a peak charge~\cite{ElYaniv:01,Zhou:08}. In particular, the classic competitive guarantees in \owt depend on the price fluctuation ratio $U/L$~\cite{ElYaniv:01}; the optimal \owt competitive ratio is unbounded as $U/L\to\infty$, which occurs when $L\to0$ for fixed $U>0$, ruling out any finite competitive ratio in the setting without a peak charge. In \paowt, we allow $L=0$, as the peak charge forces the offline optimal strategy to incur a minimum peak charge $C/T$, and this is structure not present in \owt. 
Prices can otherwise vary arbitrarily within the range $[L, U]$, and thus the set of possible price instances is $\Omega_T=[L,U]^T$. We assume neither a stochastic price model nor access to future prices.

For one-way trading without a peak charge, prior work has considered both a \textit{known-horizon} setting~\cite{ElYaniv:01} where the player knows $T$ before the first time step, as well as an \textit{unknown-horizon} setting~\cite{Lin:19} where $T$ is not known beforehand and is only revealed at the start of the final (i.e., $T$th) time step to force a compulsory trade, if needed.

For the latter case of unknown horizon, we show the following impossibility result for \paowt:

\begin{restatable}{theorem}{unknownTLB}\label{thm:unknown-T-lb}
    If $L=0$, $C\le U$, and the horizon $T$ is unknown, any algorithm for \paowt has competitive ratio $\Omega(T)$.
\end{restatable}
\begin{proof}[Proof Sketch]
    Fix an arbitrary algorithm and consider $k$ rounds at price $1/k$, with $k$ large enough that $1/k\le U$. If the algorithm buys more than half a unit on some such prefix, its peak exceeds $1/(2k)$ and its cost is at least $(1+C)/(2k)$. The adversary then extends the horizon to an arbitrarily large $T$ with prices $L=0$. An offline plan buys uniformly over these $T-k$ rounds, at cost $C/(T-k)$. Hence the competitive ratio is at least $(1+C)(T-k)/(2kC)=\Omega(T)$ as $T$ grows with $k$ fixed: the algorithm cannot undo its earlier purchases or peak.

    Otherwise, the algorithm buys at most half a unit on every such prefix. The adversary announces round $k+1$ as the final round, using price $U$, forcing a purchase of at least half a unit and cost at least $U/2$. An offline plan instead buys $1/k$ in each preceding round, at total cost $(1+C)/k$. With $T=k+1$, the competitive ratio is at least $U(T-1)/(2(1+C))=\Omega(T)$. Thus, without horizon knowledge, the algorithm cannot avoid both risks: committing early leaves an irreversible cost if cheap rounds continue, while waiting risks a large compulsory trade.
\end{proof}

\noindent The full proof is in \autoref{app:prelim}; this result in particular implies that, since $T$ is unknown, the competitive ratio can be made arbitrarily large as $T \to \infty$. Moreover, as we show in the following proposition, this linear dependence on $T$ is tight in general: any purchasing strategy for \paowt has a competitive ratio $O(T)$.

\begin{proposition}\label{prop:unknown-T-ub}
    If $L=0$, any algorithm for \paowt has competitive ratio at most
    $$\frac{U+C}{C/T} = \frac{T(U+C)}{C} = O\left( T \right)$$
\end{proposition}
\begin{proof}
This follows directly from the assumptions: any optimal solution must incur a purchasing cost of at least $L$ (where $L=0$), and a peak charge of at least $C/T$.  The worst-case total cost for any purchasing is at most $U+C$, giving the stated competitive ratio.
\end{proof}

Due to the pessimistic bounds in \autoref{thm:unknown-T-lb} and \autoref{prop:unknown-T-ub} in the unknown horizon setting, we make the following assumption:
\begin{assumption} \label{asm:knownT}
    Time horizon $T$ is known in advance.
\end{assumption}
This assumption is well motivated by target applications such as wholesale electricity purchasing, where the billing horizon is known in advance and typically spans one month.
Finally, it is clear that, for very large $C$, the peak charge will dominate the total cost, making it optimal to focus on simply minimizing the peak charge. In particular, we have the following lemma:

\begin{restatable}{lemma}{LargeCLB} \label{lem:largeC}
    When $C\ge (T-1)(U-L)$, uniform purchasing $x_t=\nicefrac{1}{T}$ for all $t$ is optimal for \paowt.
\end{restatable}

\noindent The full proof is in \autoref{app:prelim}. Since a very large $C$ makes the problem trivial, we make the following assumption:

\begin{assumption} \label{asm:smallC}
   The peak charge price $C$ satisfies $C < (T-1)(U-L)$.
\end{assumption}

\noindent For the rest of the paper, we denote $\kappa:=\frac{U}{L+\frac{C}{T}}$. Here, $L+\frac{C}{T}$ is a lower bound on the cost of any feasible solution: purchasing one unit costs at least $L$, and the peak charge is at least $\frac{C}{T}$. This parameter plays a role analogous to the price fluctuation ratio $\frac{U}{L}$ in \owt.

%% file: sec-robust-alg.tex
In this section, we present a dual-threshold algorithm called \textbf{P}eak-\textbf{A}ware \textbf{W}ater Level \textbf{S}moothing (\rob) 
and analyze its competitive ratio.
We prove that \rob achieves a competitive ratio at most $2\sqrt{\kappa}$; a matching lower bound in \autoref{sec:lower-bound} establishes its order-optimality.

\subsection{Designing Peak-Aware Water Level Smoothing}
\rob (\sref{Algorithm}{alg:rob}) is designed to optimally balance the purchase cost and the peak charge. Our algorithm is inspired by the structural properties of optimal primal and dual solutions to the offline version of \paowt.
Specifically, when all prices are known, the formulation of \paowt in~\eqref{eq:paowt-model} can be rewritten as follows:
\begin{align}
\label{eq:offline-primal}
\min_{x_t,\, v\ge 0} \quad & \sum_{t\in[T]} p_t x_t + C v, \quad \text{s.t.} \sum_{t\in[T]} x_t \ge 1 \text{ and } x_t \le v, \ \forall t\in[T].
\end{align}
The dual of this problem introduces a multiplier $\beta_t\ge0$ for each constraint $x_t\le v$ and a multiplier $\lambda\ge0$ for the requirement to purchase at least one unit; this yields the dual problem:
\begin{align}
\label{eq:offline-dual}
    \max_{\beta_t,\,\lambda\ge 0} \quad & \lambda \quad \text{s.t.} \sum_{t\in[T]} \beta_t \le C \text{ and } p_t + \beta_t \ge \lambda, \ \forall t\in[T]
\end{align}
Let $\{x_t^*\}_{t\in[T]}$ and $v^*$ denote an optimal primal solution to~\eqref{eq:offline-primal}, and let $\{\beta_t^*\}_{t\in[T]}$ and $\lambda^*$ denote an optimal dual solution to~\eqref{eq:offline-dual}. These solutions can be characterized as follows:
\begin{lemma}
\label{lem:opt-structure}
The optimal dual variable $\lambda^*$ in~\eqref{eq:offline-dual} satisfies
\begin{align}
\label{eq:opt-lambda}
\sum\nolimits_{t\in[T]} (\lambda^* - p_t)_+ = C,
\end{align}
and $\beta_t^* = (\lambda^* - p_t)_+$ for all $t\in[T]$. An optimal primal solution to~\eqref{eq:offline-primal} is given by
\begin{align}\label{eq:offline-opt-solution}
x_t^* =
\begin{cases}
 \frac{1}{k}, & p_t \le \lambda^*,\\
 0, & \text{otherwise},
\end{cases}
\qquad
v^* = \frac{1}{k},
\end{align}
where $k:=|\{t\in[T]:p_t\le\lambda^*\}|$ is the number of time slots whose prices are no greater than $\lambda^*$.
\end{lemma}
The offline primal and dual solutions to \paowt have a different structure from those of the traditional one-way trading problem~\cite{ElYaniv:01}, in which an optimal primal solution purchases one unit in a single time step at the lowest price. Due to the peak charge, the optimal primal solution characterized above purchases only in time steps whose prices are no greater than the optimal dual price $\lambda^*$, with an equal amount $1/k$ purchased in each such slot. The optimal dual solution has a water-filling structure, with $\lambda^*$ serving as the water level.
However, determining the optimal water level $\lambda^*$ is challenging in the online setting. Even if $\lambda^*$ were known \emph{a priori}, one would still not know in advance the number of time slots in which to purchase, i.e., with prices at or below this level. 
Thus, there are two key challenges in designing an online algorithm for this problem: first, we must determine an appropriate estimate of the optimal water level, to obtain prices that are competitive with the optimal value; and second, we must decide how many times to purchase, so that we do not either significantly overpay on the peak charge, or miss out on a small number of good prices. 

Our algorithm \rob exploits this structure by setting a target water level $\Phi$ and using it to guide its purchasing decisions. In particular, at each time step $t$, we aim to purchase
\begin{align}
\label{eq:purchase}
    d_t = \frac{(\Phi - p_t)_+}{C},
\end{align}
provided that the remaining amount to purchase $r$ is at least $d_t$ and no compulsory purchase is needed.
This target purchase amount has two useful properties. First, if $\Phi = \lambda^* = \C{\Opt}$, then $\sum_{t\in[T]}(\Phi-p_t)_+/C=1$, so purchasing $d_t$ in every time slot would yield exactly one unit in total. Second, $d_t$ is larger when $p_t$ lies farther below the target water level $\Phi$, so the algorithm aims to purchase more at lower prices.

One issue with the purchase decision in~\eqref{eq:purchase} arises when $\Phi < \lambda^*$. In this case, $\sum_{t\in[T]} d_t < 1$, so compulsory purchases near the end of the time horizon may lead to a large peak charge. To avoid this, we set $m = \frac{\Phi-L}{C} \ge d_t , \forall t\in[T]$ as the purchase limit. At time $t$, if the remaining amount to purchase $r$ satisfies $r - (T-t) m > \min\{r,d_t\}$, the compulsory purchase phase begins with purchase $r - (T-t) m$, and then the algorithm purchases $m$ in each remaining time slot to ensure that the peak purchase does not exceed $m$.

In our main results, \rob sets the target price $\Phi$ to $\sqrt{U(L+\frac{C}{T})}$ (see the proof of \autoref{thm:rob-ub} for the intuition behind this choice).
Note that as $C \to 0$ (eliminating the peak charge), the rate limiter is non-binding (i.e., $m=1$, allowing \rob to purchase the entire unit in a single time step).

\begin{algorithm}[t]
\caption{$\rob_\Phi$ Algorithm}\label{alg:rob}
\begin{algorithmic}[1]
\Require target water level $\Phi$ 
\Ensure online decisions $\{x_t\}_{t\in[T]}$
\State \textbf{initialize:}
\State \hspace*{\algorithmicindent}\parbox[t]{.6\linewidth}{\raggedright $m \gets \min\{1, \frac{\Phi-L}{C}\}$} \Comment{Purchase Cap}
\State \hspace*{\algorithmicindent}\parbox[t]{.6\linewidth}{\raggedright $r \gets 1$} \Comment{Remaining Amount to Purchase}
\For{$t =1,\dots,T$}
\State$x_t \gets \max\big\{\min\{r, \frac{(\Phi-p_t)_+}{C}\}, r-(T-t)m\big\}$ 
\State $r \gets r-x_t$
\EndFor
\end{algorithmic}
\end{algorithm}

\subsection{Competitive Analysis}
We now prove an upper bound on the competitive ratio of \rob. The analysis determines the target water level $\Phi$ and the corresponding rate limiter $m$ based on problem parameters.

\begin{theorem}\label{thm:rob-ub}
    \rob is $2\sqrt{\kappa}$-competitive for \paowt when $\Phi=\sqrt{U(L+\frac{C}{T})}$.
\end{theorem}

\begin{proof} For any instance $\mathcal{I}$ of the \paowt problem, we denote the cost of \rob using target price $\Phi$ by $\C{\rob_\Phi(\mathcal{I})}$. We first establish the offline comparison used when a compulsory trade occurs. For any $\Phi\ge0$, set $F(\Phi)=\sum_{t=1}^T(\Phi-p_t)_+$ and let $v=\max_t x_t$ for a feasible purchase vector $\mathbf{x}$. If $F(\Phi)\le C$, then
\begin{equation}
    \sum_{t=1}^T p_t x_t+Cv
    \ge \Phi\sum_{t=1}^T x_t+\bigl(C-F(\Phi)\bigr)v
    \ge \Phi.
    \label{ineq:threshold-opt-lb}
\end{equation}
Thus $\C{\Opt(\mathcal{I})}\ge\Phi$ whenever $F(\Phi)\le C$. We now consider two cases.

\paragraph{No Compulsory Trade} $\rob_\Phi$ is able to purchase its entire constraint at price at most $\Phi$ with peak charge at most $Cm$. By construction of $m$, $Cm = \min\{C,\Phi-L\}$, so $\C{\rob_\Phi(\mathcal{I})} \le \Phi+\min\{C,\Phi-L\}$. Since every feasible solution buys at least one unit at prices at least $L$ and has peak purchase at least $1/T$, $\C{\Opt(\mathcal{I})}\ge L+C/T$.

\paragraph{Compulsory Trade} If a compulsory trade occurs, $\rob_\Phi$ can be forced to purchase at a price as high as $U$, but its peak charge remains at most $Cm$. Therefore, $\C{\rob_\Phi(\mathcal{I})}\le U+\min\{C,\Phi-L\}$. To bound \Opt, consider the first round $t$ in which the compulsory term $r-(T-t)m$ exceeds the threshold purchase $(\Phi-p_t)_+/C$. All earlier purchases then follow the threshold, totaling $1-r$. When $m<1$, every later threshold purchase is at most $m$, and hence
\[
    \frac{F(\Phi)}{C}
    \le (1-r)+\frac{(\Phi-p_t)_+}{C}+(T-t)m <1.
\]
When $m=1$, a compulsory trade can occur only at $t=T$, so the same inequality holds without a future term. It follows from \eqref{ineq:threshold-opt-lb} that $\C{\Opt(\mathcal{I})}\ge\Phi$.

Combining these results, we see that the competitive ratio is at most

$$\frac{\C{\rob_\Phi(\mathcal{I})}}{\C{\Opt(\mathcal{I})}} \le \max \left\{ \frac{\Phi+\min\{C,\Phi-L\}}{L+\frac{C}{T}}, \frac{U+\min\{C,\Phi-L\}}{\Phi}\right\}$$

\noindent To balance the two cost bounds, we set $\Phi=\sqrt{U(L+\frac{C}{T})}$. Since $\Phi \ge \min \{C, \Phi-L\}$, we have:

\begin{align*}
    \frac{\C{\rob_\Phi(\mathcal{I})}}{\C{\Opt(\mathcal{I})}} &\le \max \left\{\frac{2 \Phi}{L+\frac{C}{T}}, 1+\frac{U}{\Phi}\right\},\\
    &\le \max \left\{2 \sqrt{\frac{U}{L+\frac{C}{T}}}, 1+\sqrt{\frac{U}{L+\frac{C}{T}}}\right\},\\
    &\le  2\sqrt{\frac{U}{L+\frac{C}{T}}}.
\end{align*}
Thus, \rob is $2\sqrt{\kappa}$-competitive when $\Phi=\sqrt{U(L+\frac{C}{T})}$ and $m=\min\{1,\frac{\Phi-L}{C}\}$.
\end{proof}

%% file: sec_preliminaries.tex
The preceding section gives an online algorithm with competitive ratio at most $2\sqrt{\kappa}$. 
Under the assumptions in \autoref{sec:prelims}, we ask the question: 
\emph{can any online algorithm achieve an asymptotically better competitive ratio?} 
In this section, we show that the answer is no:  any online algorithm for \paowt has competitive ratio $\Omega(\sqrt{\kappa})$. The proof uses two price sequences with a common prefix, forcing a trade-off between buying before prices rise and waiting for prices to fall.

\begin{theorem}\label{thm:known-T-lb}
    In the \paowt setting, under the assumptions~\ref{asm:priceBounds}-\ref{asm:smallC}, any online algorithm has competitive ratio at least 
    $\max\big\{1,\frac{1}{2\sqrt{3}}\sqrt{\kappa}\big\}$.
\end{theorem}
\begin{proof}
    The lower bound of 1 comes directly from the definition of the competitive ratio. We now consider the $\frac{1}{2\sqrt{3}}\sqrt{\kappa}$ case. 

    The argument uses two price sequences with a common prefix. Because an online algorithm cannot see future prices, any algorithm $\alg$ will make the same purchasing decisions for the common prefix of both sequences, resulting in a performance trade-off across the suffixes. Let $n = \lceil \frac{T}{2}\rceil$, $k = \lfloor \frac{T}{2}\rfloor$, and define sequences like so:
    
    $$\mathcal{I}_1 = \underbrace{p,\dots,p}_{n},\underbrace{U,\dots,U}_{k} \quad \quad \mathcal{I}_2 = \underbrace{p,\dots,p}_{n},\underbrace{L,\dots,L}_{k}$$
    where $p = \sqrt{U\delta} \in[L,U]$, for $\delta = L+\frac{C}{k}$. Note that $\frac{C}{n} \le \delta$ and $\delta\le 3(L+\frac{C}{T})$. We consider two cases based on the relative size of $\delta$ and $U$.

    \paragraph{Case: $\delta \le U$} Suppose \Alg purchases a total volume of $s$ across the prefix of $n$ prices $p$. 
    
    On $\mathcal{I}_1$, purchasing $1/n$ in each of the first $n$ rounds shows that $\C{\Opt}\le p+C/n$, while any feasible \Alg pays at least $ps+U(1-s)$. On $\mathcal{I}_2$, purchasing $1/k$ in each of the final $k$ rounds shows that $\C{\Opt}\le L+C/k=\delta$, while \Alg pays at least $ps$. These feasible offline schedules suffice; they need not be optimal.

    $$\frac{\C{\Alg}}{\C{\Opt}} \ge \min\limits_{0\le s\le 1} \max\left\{\frac{ps}{\delta},\frac{U-(U-p)s}{p+C/n}\right\} = \frac{U}{2\sqrt{U\delta}+C/n-L-C/k}$$

    Since $\nicefrac{C}{n}\le \delta \le 3(L+\nicefrac{C}{T})$, we can lower bound the expression by 

    $$\frac{\C{\Alg}}{\C{\Opt}} \ge \frac{U}{2\sqrt{U\delta}}\ge \frac{1}{2\sqrt{3}}\sqrt{\frac{U}{L+C/T}}$$
 
    \paragraph{Case: $\delta > U$} Here we can ignore the price sequences and recognize the premise implies that $\nicefrac{U}{(L+C/T)} \le 3$. Thus
    $$\frac{1}{2\sqrt{3}}\sqrt{\frac{U}{L+C/T}}\le \frac{1}{2}$$

    Since the competitive ratio is at least 1, it is lower bounded by $\frac{1}{2\sqrt{3}}\sqrt{\frac{U}{L+C/T}}$.
\end{proof}

Compared to \owt (without a peak charge), the lower bound for \paowt intuitively depends on the additional peak charge $C$.  As $C \rightarrow 0$, the \paowt lower bound smoothly recovers the optimal $\Omega\big(\sqrt{{U}/{L}}\big)$ competitive bound for the minimization version of \owt~\cite{Lechowicz:24OCS, Lorenz:09}.

Together, \autoref{thm:rob-ub} and \autoref{thm:known-T-lb} establish that the optimal competitive ratio for \paowt is $\Theta(\sqrt{\kappa})$, attained by \rob up to constant factors.

%% file: sec_learning-aug-algs.tex
To go beyond the worst-case analysis, in this section, we aim to further investigate how to make use of machine-learned predictions to design learning-augmented algorithms for \paowt. 
Suppose that an algorithm receives predictions of the form $\adv = (\hat{x}_1,\dots, \hat{x}_T)$, where $\hat{x}_t \in [0,1]$ denotes a prediction of the optimal purchase \textit{decision} at time $t$.
We aim to design an algorithm that uses the decision advice $\adv$ to improve the average-case performance for \paowt when $\adv$ is of good quality, while retaining a worst-case bound.

To investigate this, we evaluate learning-augmented algorithms~\cite{Lykouris:18, Purohit:18} using two metrics: \textit{consistency}, which compares the algorithm's cost with that of the decision advice; and \textit{robustness}, which bounds its competitive ratio regardless of the quality of the advice.
Given a feasible advice $\adv$ for instance $\mathcal{I}$, we let $\Alg(\mathcal{I},\adv)$ denote the online decision of learning-augmented algorithm. Given a parameter $\epsilon > 0$, we say the algorithm is $(1+\epsilon)$-consistent and $\gamma(\epsilon)$-robust if for every instance $\mathcal{I}$ and feasible advice $\adv$,
\begin{align*}
    &\C{\Alg(\mathcal{I},\adv)} \le (1+\epsilon) \C{\adv},\\ 
    &\C{\Alg(\mathcal{I},\adv)} \le \gamma(\epsilon) \C{\Opt(\mathcal{I})}.
\end{align*}
In particular, when $\adv$ is the optimal solution, the consistency guarantees implies a competitive ratio of at most $1+\epsilon$.

\subsection{Lower Bound on Consistency-Robustness Trade-off}

We begin by proving a result that gives a fundamental limit on the trade-off between consistency and robustness. In particular, we have the following lower bound:

\begin{theorem}\label{thm:advice-gen-lb}
    Any learning-augmented algorithm for \paowt that is $(1+\epsilon)$-consistent must be at least $\gamma(\epsilon)$-robust, where $\gamma(\epsilon) \ge \max\left\{1,\frac{1}{2\sqrt{3}}\sqrt{\kappa},\frac{\kappa}{12(1 + \epsilon)}\right\}$.
\end{theorem}

\begin{proof}
    The first two terms follow from \autoref{thm:known-T-lb}: with arbitrary predictions, a learning-augmented algorithm cannot guarantee a better competitive ratio than the prediction-free lower bound. We therefore prove the remaining bound, $\gamma\ge \frac{U}{12(1+\epsilon)(L+\frac{C}{T})}$.

    If $\frac{U}{L+\frac{C}{T}}<4(1+\epsilon)$, this term is less than $\frac{1}{3}$ and is already covered by the first two terms. Otherwise, suppose $\frac{U}{L+\frac{C}{T}}\ge4 (1+\epsilon)$. We construct two instances with the same prices and advice during an initial block of rounds. On the first instance, the advice is optimal, so consistency requires the algorithm to buy a substantial amount during this block. On the second, prices fall after the block, but the algorithm has already made the same purchases.

    Let $n=\lceil \frac{T}{2}\rceil$, $k=\lfloor \frac{T}{2}\rfloor$, and $p=\frac{U}{2(1+\epsilon)}-\frac{C}{n}$. Since $(1+\epsilon)\ge1$ and $n\ge \frac{T}{2}$, the assumption that $\frac{U}{L+\frac{C}{T}}\ge4(1+\epsilon)$ gives
    \[
        p\ge 2\left(L+\frac{C}{T}\right)-\frac{C}{n}\ge L,
        \qquad p\le \frac{U}{2(1+\epsilon)}\le U.
    \]
    Thus all prices in the following instances are admissible:
    \[
        \mathcal{I}_1=
        \underbrace{p,\ldots,p}_{n\text{ }}, 
        \underbrace{U,\ldots,U}_{k\text{ }}
        \qquad
        \mathcal{I}_2=
        \underbrace{p,\ldots,p}_{n\text{ }}, 
        \underbrace{L,\ldots,L}_{k\text{ }}
    \]
    On both instances, a learning-augmented algorithm receives the same decision advice: buy $\frac{1}{n}$ in each of the first $n$ rounds and nothing thereafter. Let $z$ denote its total purchases during the shared initial block. Because the algorithm sees the same prices and advice before the last $k$ rounds, this value of $z$ is the same on both instances.

    \paragraph{Instance 1: correct advice forces early purchases.}
    Following the advice on $\mathcal{I}_1$ costs $p+\frac{C}{n}=\frac{U}{2(1+\epsilon)}$. This schedule is optimal, not just feasible. To see why, consider any schedule buying exactly one unit, with volume $s$ in the first $n$ rounds. Its peak is at least $\frac{s}{n}$, so its total cost is at least
    \[
        ps+U(1-s)+\frac{Cs}{n}
        =\left(p+\frac{C}{n}\right)s+U(1-s)
        \ge p+\frac{C}{n}.
    \]
    The last inequality holds because $p+\frac{C}{n}=\frac{U}{2(1+\epsilon)}\le U$. Since prices and the peak charge are nonnegative, buying more than one unit cannot improve the optimum. Hence $\C{\Opt(\mathcal{I}_1)}=\frac{U}{2(1+\epsilon)}$, and $(1+\epsilon)$-consistency requires
    \[
        \C{\Alg(\mathcal{I}_1)}
        \le (1+\epsilon)\,\C{\Opt(\mathcal{I}_1)}
        =\frac{U}{2}.
    \]
    But the algorithm must buy any volume still missing after the initial block at price $U$, giving $\C{\Alg(\mathcal{I}_1)}\ge U(1-z)$. Combining these inequalities yields $z\ge\frac{1}{2}$.

    \paragraph{Instance 2: the same purchases are costly when prices fall.}
    On $\mathcal{I}_2$, the algorithm again buys $z\ge\frac{1}{2}$ during the initial block, even though the final $k$ rounds now offer price $L$. These purchases alone cost $pz$ and establish a peak of at least $\frac{z}{n}$. Therefore,
    \[
        \C{\Alg(\mathcal{I}_2)}
        \ge pz+\frac{Cz}{n}
        =\left(p+\frac{C}{n}\right)z
        \ge \frac{U}{4(1+\epsilon)}.
    \]
    An offline plan can instead wait and buy $\frac{1}{k}$ in each of the final $k$ rounds. This gives the upper bound
    \[
        \C{\Opt(\mathcal{I}_2)}
        \le L+\frac{C}{k}
        \le 3\left( L+\frac{C}{T} \right),
    \]
    where the last inequality uses $\frac{T}{k}\le3$ for every $T\ge2$. Robustness must cover this instance even when the advice is incorrect, so
    \[
        \gamma
        \ge \frac{\C{\Alg(\mathcal{I}_2)}}{\C{\Opt(\mathcal{I}_2)}}
        \ge \frac{U}{12(1+\epsilon)(L+\frac{C}{T})} = \frac{\kappa}{12(1+\epsilon)}.
    \]
    Thus, stronger consistency forces more commitment to the shared initial block, limiting robustness when the same advice is wrong.
\end{proof}

Using the optimal consistency-robustness trade-off shown in \autoref{thm:advice-gen-lb} as a benchmark, we next consider how an algorithm can use predictions to improve performance.

\subsection{Naive Convex Combination Algorithm}

Given decision advice $\adv$, a natural learning-augmented algorithm is to apply a linear combination of some robust decisions (e.g., from \rob) and the decision advice~\cite{Lechowicz:24OCS}. Given a linear combination parameter $\mu \in [0,1]$, we define an algorithm \mix that combines robust decisions from \rob and the decision advice in \sref{Algorithm}{alg:mix}.

\begin{algorithm}
\caption{\mix Algorithm}\label{alg:mix}
\begin{algorithmic}[1]
\Require consistency parameter $\epsilon\ge0$, $\alpha$-competitive algorithm $\rob_\Phi$
\Ensure online decisions $\{x_t\}_{t\in[T]}$
\State \textbf{initialize:}
\State \hspace*{\algorithmicindent}\parbox[t]{.6\linewidth}{\raggedright $r \gets 1$} \Comment{Remaining amount to purchase}
\State \hspace*{\algorithmicindent}\parbox[t]{.6\linewidth}{\raggedright $\mu \gets 1 - \frac{\epsilon}{\alpha - 1}$}
\Comment{Setting combination parameter}
\For{$t =1,\dots,  T$}
\State observe price $p_t$ and advice $\hat{x}_t$
\State $\bar{x}_t \gets \texttt{$\rob_\Phi$\_Allocation}(r,t)$
\State $x_t \gets \mu \hat{x}_t+(1-\mu)\bar{x}_t$ 
\EndFor
\end{algorithmic}
\end{algorithm}

We now analyze the consistency and robustness of \mix:

\begin{restatable}{theorem}{mixAlg}\label{thm:mixAlg}
For a given $\epsilon > 0$, \mix is $(1 +\epsilon)$-consistent and $\big((1 - \frac{\epsilon}{\alpha - 1})\kappa (1 + \frac{C}{U}) + \frac{\epsilon\alpha}{\alpha - 1}\big)$-robust. %
\end{restatable}
Note that when $\epsilon \to 0$, the robustness of the \mix is approaching $\kappa (1 + \frac{C}{U})$, where $\frac{C}{U}$ can be of order $O(T)$. 
Therefore, simple \texttt{MIX} cannot achieve the consistency-robustness lower bound shown in \autoref{thm:advice-gen-lb}. 
\begin{proof}
First, the online solutions $\{\bar{x}\}_{t\in[t]}$ from $\rob_\Phi$ and the advice $\{\hat{x}_t\}_{t\in[T]}$ are feasible solutions, and thus the solution $x_t = \mu \hat{x}_t + (1 - \mu \bar{x}) $ from \texttt{MIX} is also feasible, i.e.,
\begin{align*}
    \sum_{t\in[T]} x_t = \mu \sum_{t\in[T]}  \hat{x}_t + (1 - \mu) \sum_{t\in[T]} \bar{x}_t \ge \mu \cdot 1 + (1 - \mu) \cdot 1 = 1.
\end{align*}

The total cost of the \texttt{MIX} is 
\begin{align*}
\C{\texttt{MIX}} &= \sum_{t\in[T]} p_t x_t + C \max_{t\in[T]} x_t\\
&\le \mu \sum_{t\in[T]} p_t \hat{x}_t + (1 - \mu) \sum_{t\in[T]} p_t \bar{x}_t + \mu C \max_{t} \hat{x}_t + (1 - \mu) C \max_{t} \bar{x}_t\\
& \le \mu \C{\adv} + (1 - \mu) \alpha \cdot \C{\Opt}.\\
\end{align*}

Since $\C{\Opt} \le \C{\adv}$, we have $\C{\texttt{MIX}} \le [\mu + (1 - \mu) \alpha] \cdot \C{\adv } = (1 + \epsilon) \C{\adv}$. Thus, \texttt{MIX} is $(1+\epsilon)$-consistent. 

In addition, note that $\C{\adv} \le U + C$ and $\C{\Opt} \ge L + C/T$, then $\C{\adv} \le \frac{U+C}{L+C/T} \C{\Opt }$. Thus, $\C{\texttt{MIX}} \le [\mu \frac{U+C}{L+C/T} + (1 - \mu) \alpha] \cdot \C{\Opt } = ((1 - \frac{\epsilon}{\alpha - 1})\frac{U+C}{L+C/T} + \frac{\epsilon\alpha}{\alpha - 1}) \C{\Opt}$. Thus, the robustness of \texttt{MIX} is $(1 - \frac{\epsilon}{\alpha - 1})\frac{U+C}{L+C/T} + \frac{\epsilon\alpha}{\alpha - 1} = (1 - \frac{\epsilon}{\alpha - 1})\kappa (1 + \frac{C}{U}) + \frac{\epsilon\alpha}{\alpha - 1}$.
\end{proof}

Therefore, to achieve consistency below $\sqrt{\kappa}$, the \mix algorithm requires robustness with a linear dependence on $T$ (and $C$) which is not required for general learning-augmented algorithms.
In order to match the lower bound of \autoref{thm:advice-gen-lb}, we next propose a more sophisticated algorithm that carefully controls the impact of the peak charge when provided with incorrect decision advice.

\subsection{Augmenting \rob with Decision Advice}

In this section, we design an algorithm called \textbf{P}eak-\textbf{A}ware \textbf{W}ater Level \textbf{S}moothing with \textbf{D}ecision \textbf{A}dvice (\decision), which augments \rob with \textit{advice-tracking decisions} to match a specific consistency bound $1+\epsilon$. 
Depending on the value of $\epsilon$, \decision calculates a method of following the advice in one of two ways.

\paragraph{Near 1-consistency}
The algorithm will \textit{aggressively track} advice when $0 \le  \epsilon < 1$. \decision aims to match the accumulated volume purchased by the advice stream, with per-step purchases capped by a rate limiter $m$. In this setting, the algorithm does not consider the arriving prices, only the advice stream. Although the rate limiter may lead to added price costs compared to \mix, when advice is accurate, we show that the savings to the peak charge outweigh them. In addition, the rate limiter allows \decision to avoid the linear dependence on $T$ seen in \mix.  The full procedure is given in \autoref{alg:decision}.

\paragraph{General consistency} When $\epsilon\ge 1$, the algorithm engages in \textit{balanced tracking} of advice. In this setting, similar to \rob, \decision takes a target water level parameter $\Phi$ that influences purchases via the water-filling solution $d_t = \nicefrac{(\Phi-p_t)_+}{C}$, and a rate limiter parameter $m$.

\decision also takes a parameter $H>\max\{\nicefrac{C}{T}, \Phi-L\}$ that allows it to decide how to follow the decision advice.  At a high level, \decision improves upon \mix by explicitly accounting for the accrued cost of the decision advice so far: when the running cost of advice is high, \decision switches to ignore it to preserve robustness.  
In particular, $H$ bounds the cost attributable to the advice; it does so by influencing a set of shadow decisions $y_t$, where a larger $H$ corresponds to more willingness to follow a large $\hat{x}_t$ (at the risk of incurring an ill-advised large peak charge).  Reflecting this intuition, $m$ is set to $\nicefrac{H}{C} \ge \nicefrac{1}{T}$, meaning that there is a looser restriction on the purchasing rate as $H$ grows.

In each time step, \decision computes the dual decision $d_t$ and the shadow decision $y_t$, purchasing the larger of the two.  This decision is capped by $m$ and overridden if compulsory trading is needed to meet the constraint.  The full procedure is given in \autoref{alg:decision}. We show that to reach the target consistency, it suffices to set $\Phi = L+\frac{\epsilon(L+C/T)}{2}$ and $H = \frac{U}{\epsilon}$ for $\epsilon \le \sqrt{\kappa}$. If input $\epsilon> \sqrt{\kappa}$, we instead run the algorithm with $\epsilon= \sqrt{\kappa}$ as there is no robustness benefit to higher consistency beyond this point.

\begin{algorithm}[t]
\caption{\decision Algorithm}
\label{alg:decision}
\begin{algorithmic}[1]
\Require Consistency parameter $\epsilon\ge 0$
\Ensure online decisions $\{x_t\}_{t\in[T]}$
\State \textbf{initialize:}
\State \hspace*{\algorithmicindent}\parbox[t]{.6\linewidth}{\raggedright $r \gets 1$} \Comment{Remaining amount to purchase}

\State \hspace* {\algorithmicindent}\parbox[t]{.6\linewidth}{\raggedright \textbf{if} $\epsilon < 1 $ \textbf{then}}
\State \hspace*{\algorithmicindent+\algorithmicindent}\parbox[t]{.6\linewidth}{\raggedright $m \gets \min\{1,(U-L)/C\}$} 
\State \hspace*{\algorithmicindent+\algorithmicindent}\parbox[t]{.4\linewidth}{\raggedright $A \gets 0$} \Comment{Cumulative advice's purchase}
\State \hspace*{\algorithmicindent+\algorithmicindent}\parbox[t]{.4\linewidth}{\raggedright $X\gets 0$} \Comment{Cumulative \decision's purchase}
\State \hspace* {\algorithmicindent}\parbox[t]{.6\linewidth}{\raggedright \textbf{else}}
\State \hspace*{\algorithmicindent+\algorithmicindent}\parbox[t]{.6\linewidth}{\raggedright $\epsilon\gets \min\{\epsilon,\sqrt{\kappa}\}$} 
\State \hspace*{\algorithmicindent+\algorithmicindent}\parbox[t]{.6\linewidth}{\raggedright $\Phi \gets L+\frac{\epsilon(L+C/T)}{2}$} 
\State \hspace*{\algorithmicindent+\algorithmicindent}\parbox[t]{.6\linewidth}{\raggedright $H \gets \frac{U}{\epsilon}$} 
\State \hspace*{\algorithmicindent+\algorithmicindent}\parbox[t]{.6\linewidth}{\raggedright $m \gets \min\{1,H/C\}$} 
\State \hspace*{\algorithmicindent+\algorithmicindent}\parbox[t]{.4\linewidth}{\raggedright $(y_1,\ldots,y_T)\gets (0,\ldots,0)$} \Comment{Shadow decision variables}
\State \hspace* {\algorithmicindent}\parbox[t]{.6\linewidth}{\raggedright \textbf{end if}}

\\
\For{$t=1,\ldots,T$}
    \State Observe price $p_t$ and advice $\hat{x}_t$
    \State $\ell_t\gets \max\{0,r-(T-t)m\}$
        \Comment{Deadline-feasibility floor}
    \State $u_t\gets \min\{r,m\}$
        \Comment{Remaining-demand and rate cap}

    \If{$\epsilon <1$}
        \State $A\gets A+\hat{x}_t$
        \State $z_t\gets A-X$
            \Comment{Unpurchased advice target}
    \ElsIf{$\epsilon \ge 1$}
        \State $y_t \gets \max\limits_{y_t\in[0,\hat{x}_t]}:\sum\limits_{s\le t} p_s y_s+C\max\limits_{s\le t} y_s \le H$ \Comment{Advice-driven target}
        \State $d_t\gets (\Phi-p_t)_+/C$
            \Comment{Price-driven target}
        \State $z_t\gets \max\{y_t,d_t\}$
    \EndIf

    \State $x_t\gets \min\{u_t,\max\{\ell_t,z_t\}\}$
        \Comment{Project target onto the feasible interval}
    \State $r\gets r-x_t$
    \State $X\gets X+x_t$
\EndFor
\end{algorithmic}
\end{algorithm}

\begin{restatable}{theorem}{DecisionConsistency}\label{thm:decision-advice-consistency}
Given input parameter $\epsilon \ge 0$, \decision is $(1+\epsilon)$-consistent.

\end{restatable}

\begin{proof}
    We give separate arguments for  $\epsilon <1$ (aggressive tracking) and $\epsilon \ge 1$ (balanced tracking).

    \paragraph{Aggressive tracking}
    We split the analysis into cases based on the size of the largest advice decision $\max_{i\in[T]} \hat{x}_i$. If $\max_{i\in [T]} \hat{x}_i \le m$, then the advice respects the rate limiting parameter. Note that the algorithm is able to follow advice exactly, as $z_t$ will always equal $\hat{x}_t\le m$.  Because it follows the advice exactly, \decision is $1$-consistent in this case.
    
    Now, if $\max_{i\in [T]} \hat{x}_i > m$, we note that $x_t\ge \min\{\hat{x}_t,m\}$. We define the total amount of ``deferred advice'' as
    $$\sum_t(\hat{x}_t-x_t)_+ \le\sum_t(\hat{x}_t-m)_+ \le \frac{\max_{i\in [T]} \hat{x}_i - m}{\max_{i\in [T]} \hat{x}_i}.$$

    Let \adv represent following the decision advice exactly. Deferred advice increases the price cost of \decision compared to \texttt{ADV}---in the worst-case, this price increase is at most $\frac{\max_{i\in [T]} \hat{x}_i - m}{\max_{i\in [T]} \hat{x}_i}(U-L)$.  Simultaneously however, \decision's peak charge is reduced by at least $C(\max_{i\in [T]} \hat{x}_i - m)$.  We have:

    \begin{align*}
        \C{\decision}-\C{\adv}&\le\frac{\max_{i\in [T]} \hat{x}_i - m}{\max_{i\in [T]} \hat{x}_i}(U-L) - C\left(\max_{i\in [T]} \hat{x}_i - m\right)\\
        &\le \left(\max_{i\in [T]} \hat{x}_i - m\right)\left(\frac{U-L}{\max_{i\in [T]} \hat{x}_i}-C \right) \le 0
    \end{align*}

    The final inequality holds because $\max_{i\in [T]} \hat{x}_i > m$ implies $m=\frac{U-L}{C}$ and similarly $C > \frac{U-L}{\max_{i\in [T]} \hat{x}_i}$. Since $\C{\decision} \le \C{\adv}$, \decision is 1-consistent by definition. Thus, regardless of the size of $\max_{i\in [T]} \hat{x}_i$, we show 1-consistency, which also implies $(1+\epsilon)$-consistency. %

    \paragraph{Balanced tracking}
    
    We first consider the consistency as a function of $\Phi$ and $H$ (\sref{Lemma}{lem:decisionConsistent}), deferring the full proof to \autoref{app:LAA}. This result then suggests assignments for $\Phi$ and $H$ that satisfy the desired consistency bound.

    \begin{restatable}{lemma}{decisionConsistent} \label{lem:decisionConsistent}
        The \decision algorithm is $(1+\epsilon)$-consistent for $\epsilon = \max\{\frac{2(\Phi-L)}{L+C/T}, \frac{U}{H}\}$
    \end{restatable}

    Therefore, if we set $\Phi = L+\frac{\epsilon(L+C/T)}{2}$ and $H = \frac{U}{\epsilon}$ we find

    \begin{align*}
        \max\left\{ \frac{2(\Phi-L)}{L+C/T}\frac{U}{H} \right\} 
        &=  \max\left\{\frac{2(L+\frac{\epsilon(L+C/T)}{2}-L)}{L+C/T}, \frac{U}{U/\epsilon}\right\}\\
        &= \max\{\epsilon,\epsilon\} = \epsilon.
    \end{align*}
    
    \noindent Thus, we have shown \decision is $(1+\epsilon)$-consistent in both regimes.
\end{proof}

\noindent Next, we show that for a given (desired) consistency $(1+\epsilon)$, \decision is $\gamma(\epsilon)$-robust, where $\gamma(\epsilon)$ matches the optimal consistency-robustness bound from \autoref{thm:advice-gen-lb}.

\begin{theorem}\label{thm:decision-advice-ub}
    Given input parameter $\epsilon \ge 0$,
    the \decision algorithm is $\gamma(\epsilon)$-robust, where
    
    $$\gamma(\epsilon) \le \begin{cases}
        2\kappa & \epsilon< 1 \\
        4\max\{\sqrt{\kappa}, \nicefrac{\kappa}{\epsilon   }\} & \epsilon \ge 1
    \end{cases}$$
\end{theorem}

\begin{proof}
    We split the analysis by $\epsilon$.

    \paragraph{Aggressive tracking}  

    In this case we can note that \decision has cost at most $U+Cm$. Since $m=\min\{1,(U-L)/C\}$, this results in cost at most $U+\min\{C,U-L\}$. Meanwhile, \opt has a trivial lower bound of $L+C/T$.

    $$\frac{\C{\decision}}{\C{\opt}} \le \frac{U+\min\{C,U-L\}}{L+C/T}
    \le \frac{2U}{L+C/T}\le 2\kappa$$

    \paragraph{Balanced tracking}
     If $\epsilon > \sqrt{\kappa}$, we clip $\epsilon$ to $\sqrt{\kappa}$ as there is no consistency or robustness benefit to a larger value of $\epsilon$.
     We first consider robustness as a function of $\Phi$ and $H$. We then show that the specified assignment of $\Phi$ and $H$ via $\epsilon$ gives the stated bound.

    \begin{restatable}{lemma}{decisionRobust}
        The \decision algorithm is $\gamma$-robust for $\gamma = \max\{\frac{H+2\Phi-L}{L+C/T}, \frac{U+H}{\Phi}\}.$
    \end{restatable}

    Using the above result, if we set $\Phi = L+\frac{\epsilon(L+C/T)}{2}$ and $H = \frac{U}{\epsilon}$, we find robustness: 

    \begin{align*}
        \max\left\{\frac{H+2(\Phi)-L}{L+C/T}, \frac{U+H}{\Phi}\right\} &= \max \left\{\frac{\frac{U}{\epsilon}+2(L+\frac{\epsilon(L+C/T)}{2})-L}{L+C/T}, \frac{U+\frac{U}{\epsilon}}{L+\frac{(\epsilon)(L+C/T)}{2}}\right\}\\
        &= \max\left\{\frac{\kappa}{\epsilon}+\frac{L}{L+C/T}+\epsilon, \frac{\kappa(1+\frac{1}{\epsilon})}{\frac{L}{L+C/T}+\frac{\epsilon}{2}}\right\}\\
        \text{Since } \epsilon\le \sqrt{\kappa} \text{ and } \frac{\kappa}{\epsilon}&\ge\epsilon>\frac{L}{L+C/T}, \text{ the left term } \frac{\kappa}{\epsilon}+\frac{L}{L+C/T}+\epsilon\le 3\frac{\kappa}{\epsilon} \\
        &\le \max \left\{3\frac{\kappa}{\epsilon}, \frac{\kappa(1+\frac{1}{\epsilon})}{\frac{L}{L+C/T}+\frac{\epsilon}{2}}\right\}\\
        \text{Since } 1+\frac{1}{\epsilon}\le 2 \text{ and } &\frac{L}{L+C/T}\ge 0, \text{ the right term } \frac{\kappa(1+\frac{1}{\epsilon})}{\frac{L}{L+C/T}+\frac{\epsilon}{2}}\le \frac{2\kappa}{\epsilon/2}\\
        &\le \max \left\{ 3\frac{\kappa}{\epsilon}, 4\frac{\kappa}{\epsilon}\right\} \le 4\frac{\kappa}{\epsilon}
    \end{align*}

Using the clipped $\epsilon\le\sqrt{\kappa}$, we can conclude that that \decision achieves robustness $\gamma(\epsilon) \le 4\max\{\sqrt{\kappa},\nicefrac{\kappa}{\epsilon}\}$.

\end{proof}

\subsection{Converting other forms of predictions into decision advice}

Our framework is not restricted to predictions as purchase decisions. Other forms of predictions can also be incorporated by converting them into decision advice. The resulting advice can then be used by \decision without modifying the algorithm. Moreover, whenever the converted advice can attain the offline optimal cost when the predictions are accurate, the analysis of \decision also applies and the same consistency-robustness guarantee holds.

As an example, the structure of the optimal primal and dual solutions in Lemma~\ref{lem:opt-structure} suggests a form of \textit{instance advice}: $(\hat{\lambda}, \hat{k})$, where $\hat{\lambda}$ predicts the optimal dual variable $\lambda^*$, which also equals the optimal cost, and $\hat{k} \in[T]$ predicts the number of prices $p_t$ no greater than $\lambda^*$. These predictions specify both which prices to accept and how much purchase at each accepted price. 
We can convert $(\hat{\lambda}, \hat{k})$ into decision advice $\adv(\hat{\lambda}, \hat{k}): = \{\tilde{x}_t\}_{t\in[T]}$, where $\tilde{x}_t=\nicefrac{1}{\hat{k}} \mathbbm{1} \{p_t \le \hat{\lambda}\}$. To ensure feasibility under imperfect predictions, each purchase is capped at the remaining demand, and any unmet demand is purchased in the final time slot. 

\begin{restatable}{proposition}{thmInstance}
    Following advice $\adv(\hat{\lambda}, \hat{k})$ can attain the optimal offline cost when predictions $\hat{\lambda}$ and $\hat{k}$ are accurate. Additionally, neither $\hat{\lambda}$ nor $\hat{k}$ predictions alone can attain optimal cost.
\end{restatable}
    We present a sketch here, and a full proof is available in \autoref{app:LAA}.
    Recall the offline primal-dual solution in Lemma~\ref{lem:opt-structure}. Perfect knowledge of the optimal cost $\lambda^*$ imposes a ceiling on accepted prices, but does not dictate how much to buy when $p_t \le \lambda^*$ appears. Conversely, knowledge of the optimal purchase count $k$ indicates target purchase amounts, but does not suggest which prices to accept. In both cases, this uncertainty can be exploited to create instances with cost strictly larger than the offline  optimum. 
    Put together, these predictions are complementary, describing both which prices are good, and how much to purchase.

%% file: sec-experiments.tex
In this section, we evaluate the performance of the \rob and \decision algorithms against several natural baselines.  We simulate a case study of electricity procurement with a demand charge, using real-time and day-ahead prices from four U.S. electric grids.

\subsection{Experimental Setup}
We simulate an electricity billing period where a consumer (e.g. a flexible data center) must procure a total of 1 (normalized) unit of energy in the real-time market. 
We solve the offline problem (with full knowledge of prices and demand charge) using SciPy \cite{scipy2020} to report the empirical competitive ratio of each tested algorithm.

\paragraph{Data} We use historical US grid price data obtained from GridStatus~\cite{GridStatus}. We consider real-time and day-ahead locational marginal prices (15 minute increments) for four US grid regions (ISOs): CAISO (California), ISONE (New England), ERCOT (Texas), and PJM (Mid-Atlantic). For ERCOT, CAISO, and PJM, our data covers the 2020-2024 calendar years, while ISONE covers July 2022-EOY 2024, giving $\sim$600,000 total observations. Real-time prices define the price sequence $p_t$; any negative prices are clipped to 0. 

Day-ahead prices are used to generate decision advice for learning-augmented algorithms by solving for the offline optimal on those sequences. While this is not a true machine-learned predictor, it serves as a reasonable proxy for our evaluation, as day-ahead price forecasts are readily available~\cite{Qu:24}. 

In \autoref{fig:sweeps}(d) specifically, we manipulate the decision advice directly to stress-test the robustness of our learning-augmented algorithms.  In this case, we define the decision advice as a linear combination of the true optimal solution on the real-time prices (\textit{perfect advice}) and an adversarial solution that maximizes total cost (\textit{worst-case advice}); by varying a parameter, a linear combination of the two solutions captures every scenario for decision advice.  This allows us to illustrate the effect of advice quality on performance.

\paragraph{Parameters} 
For our main experiments, we randomly sample 1000 instances for each ISO and average the results. Instances are drawn as blocks of $T=96$ consecutive time steps with no missing data; we set $L=0$, $U$ as the 99th percentile price in that ISO across the entire data set, and {$C = \nicefrac{U}{2}$}. Although we set $U$ to the 99th percentile price, this only affects algorithmic thresholds and choice of $C$; real prices are allowed to exceed this value. 
In select experiments we evaluate the effect of parameters ($C, T$, etc.) on the performance of tested algorithms.  For this sensitivity analysis, we randomly sample 100 instances from each ISO for each setting. In each variant, we test 10-15 parameter values.

\paragraph{Baselines}
In addition to evaluating the performance of \rob and \decision, we consider the following four baseline strategies.  \texttt{BlindAdvice} simply follows the decision advice, making the predicted purchase decision $\hat{x}_t$ at each time step $t$. \mix selects a convex combination of advice and \rob allocations in each step, we choose a mixing parameter $\mu = 0.5$. We also implement \owt, which is the optimal threshold-based strategy for one-way trading without a peak charge; this algorithm neglects the peak charge when making purchasing decisions and demonstrates the gain of \rob in the standard competitive setting (i.e., without predictions). In contrast to \owt, which minimizes prices, we include \texttt{Uniform} purchasing to demonstrate the effect of minimizing the peak charge. To measure the performance of \decision at varying target consistencies, we introduce a trust parameter $\xi\in[0,1]$, and set $\epsilon = \kappa^{(1-\xi)/2}$. Thus, larger $\xi$ corresponds to increased trust in advice and smaller consistency bound; note that some advice is used even when $\xi=0$.

\begin{figure}[t]
    \centering
    \includegraphics[width=\linewidth]{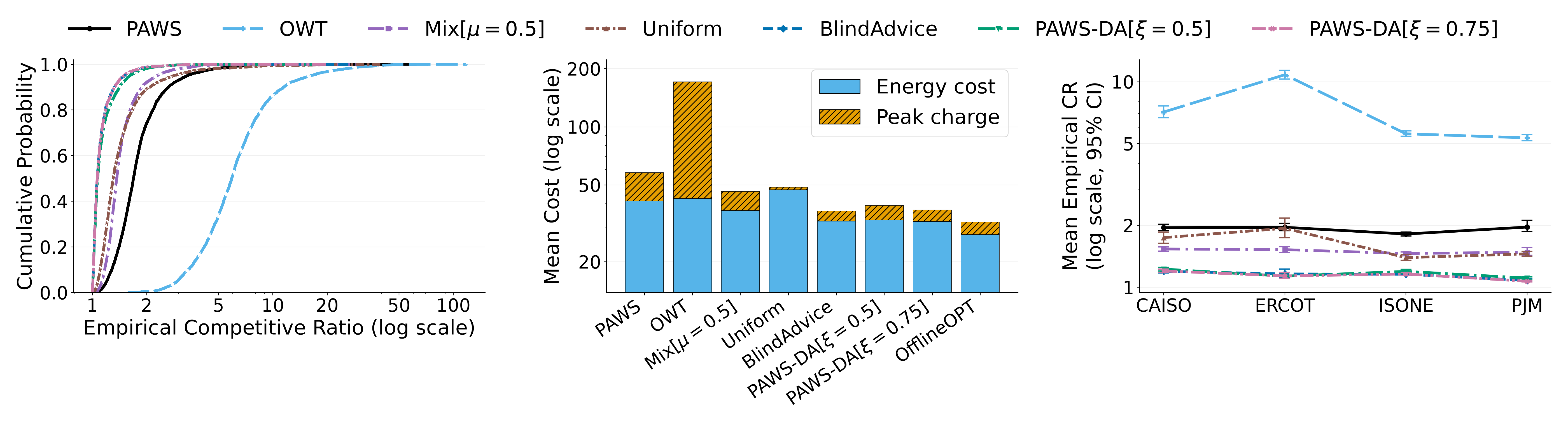}
    \makebox[0.99\textwidth][s]{ \small
        \hfill
        (a) CDF of Competitive Ratios
        \hspace{10mm}
        (b) Cost Decomposition
        \hspace{18mm}
        (c) Variation by ISO
        \hspace{6mm}
    }
    \caption{Experiments with fixed parameters. (a) Cumulative distribution function for  empirical competitive ratio for each algorithm across all ISOs. (b) Cost decomposition between energy and peak charge for each algorithm, averaged across ISO. (c) Average empirical competitive ratio and 95\% CI  by ISO. }
    \label{fig:main-results}
\end{figure}

\subsection{Experimental Results}

\paragraph{Main Results} \autoref{fig:main-results} shows the results of our main experiment setting.  In these experiments, using the day-ahead price forecasts to obtain decision advice resulted in near-optimal advice: as a result, in each plot we observe performance improving as \decision uses more of the advice. \decision[$\xi=0.75$] performs nearly identically to blind advice following, which is to be expected when advice is good.
\decision[$\xi=0.5$] exhibits performance quite close to \texttt{BlindAdvice} while retaining a strong robustness guarantee that is not present for arbitrary decision advice.
We observe that both \texttt{Uniform} and \mix[$\mu=0.5$] exhibit a better average empirical competitive ratio than \rob, but worse than the \decision variants. While both \mix[$\mu=0.5$] and \decision[$\xi=0.5$] are equally balanced on their respective consistency-robustness parameterization, the empirical competitive ratio of \decision[$\xi=0.5$] is $\approx 22\%$ lower on average. \autoref{fig:main-results}(b) shows that more of the variation between algorithms occurs due to the peak charge rather than their energy cost, suggesting that peak management is empirically more difficult. \autoref{fig:main-results}(c) demonstrates the impact of specific grids: while there is not much variation, the relatively worse performance of \owt and \texttt{Uniform} in ERCOT is unsurprising due to that grid having more frequent price spikes.

\paragraph{Sensitivity Analysis} \autoref{fig:sweeps} shows the results of our sensitivity analysis, where we vary certain parameters of the \paowt setting to observe their impact on performance. In the first three plots, high quality decision advice (based on day-ahead price forecasts) allow advice-informed algorithms to generally outperform \rob and \texttt{Uniform}. %

In \autoref{fig:sweeps}(a), all algorithms show improved performance as the relative size of $C$ increases. This result aligns with the competitive bounds of \autoref{thm:rob-ub}, which are monotonically decreasing in $C$. While the advice-informed algorithms do slightly worse than \rob when $C$ is small, they perform better than \rob at large $C$. \autoref{fig:sweeps}(b) shows the competitive ratio of all algorithms increasing in $T$. \texttt{Uniform} begins increasing immediately, while the remaining algorithms are approximately flat until $T\approx 96$. While \rob has the worst empirical competitive ratio when $T$ is small, it is less affected by increased $T$ and is nearly as strong as the advice-informed algorithms for large $T$. In contrast, \texttt{Uniform} is near-optimal for small $T$ but degrades much worse than the other baselines as $T$ increases.

In \autoref{fig:sweeps}(c) we show how the performance of \decision varies with different choices of $\xi$; as previously stated, we observe that the day-ahead forecast-based advice is near-optimal. Even with the same robustness bound the \decision[$\xi \to 0$] algorithm results in $\approx 10\%$ lower competitive ratio on real instances than the conservative \rob algorithm, though they remain within the same $95\%$ confidence interval. 

Finally, \autoref{fig:sweeps}(d) compares algorithm performance across varying advice corruption, using a linear combination of optimal and worst-case solutions across the actual prices in each instance.  We set a corruption factor between 0 and 1, where 0 corresponds to \textit{perfect decision advice} and 1 corresponds to \textit{completely adversarial advice}.
As expected, \texttt{BlindAdvice} and \mix[$\mu=0.5$] degrade linearly with advice corruption, while all \decision variants degrade sublinearly.
\rob and \texttt{Uniform} are both unaffected by advice quality.

\begin{figure}
    \centering
    \includegraphics[width=\linewidth]{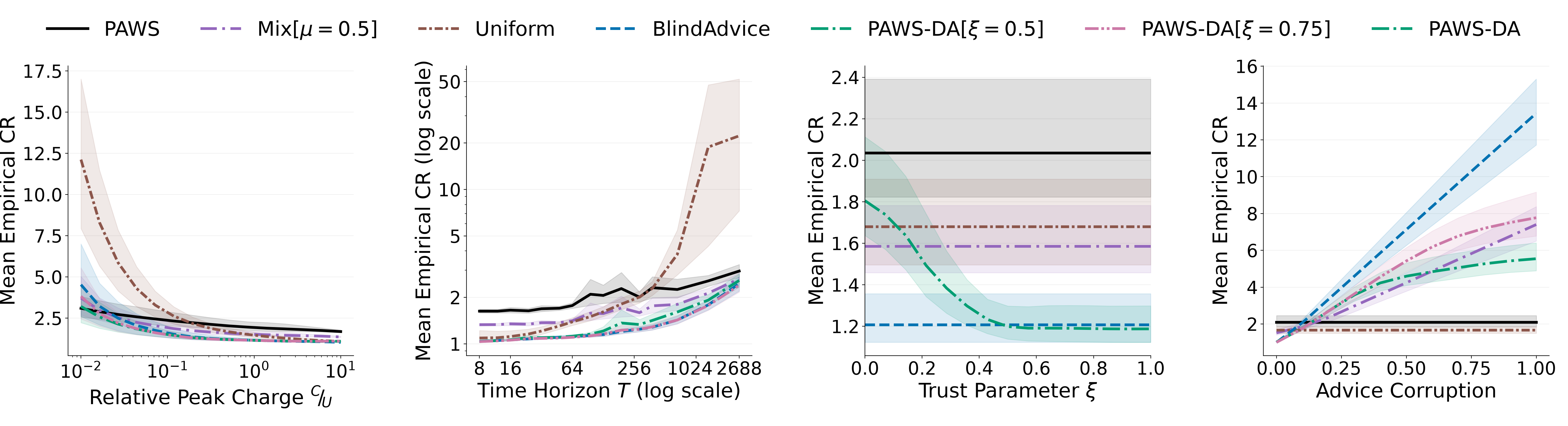}
    \makebox[0.99\textwidth][s]{\small
        \hfill
        (a) Changing C
        \hspace{14mm}
        (b) Changing T
        \hspace{15mm}
        (c) Changing $\xi$
        \hspace{11mm}
        (d) Changing Advice
        \hspace{1mm}
    }
    \caption{Experiments with varying parameters. All figures average across ISOs, showing mean empirical competitive ratio and 95\% CI. (a) Varying $C$ relative to $U$: decreasing CR as $C$ grows matches theoretical bounds. (b) Increasing time horizon $T$ makes the price-peak trade-off more challenging (c) Increasing trust parameter $\xi$ in \decision algorithm converges to blind advice. (d) Comparison of algorithms by share of adversarial advice corruption.}
    \label{fig:sweeps}
\end{figure}

%% file: sec-conclusion.tex
Motivated by electricity procurement with demand charges, we introduce and study \paowt, which extends one-way trading by additionally charging the player for the largest single purchase. This peak charge creates a tension between buying at low prices and spreading purchases across time. We characterize the limits of online purchasing and show that advance knowledge of the time horizon significantly changes what is possible in the setting. For the known-horizon setting, we design \rob, a dual-threshold algorithm that achieves an order-optimal competitive ratio. We also establish a lower bound on the consistency-robustness trade-off and design \decision, a learning-augmented algorithm that uses decision advice while providing order-optimal guarantees. We evaluate our proposed algorithms in an electricity procurement case study across four U.S. grid regions.

Several directions remain for future work. Our model considers a single purchasing obligation over a fixed horizon. A natural extension would allow multiple obligations with different deadlines that contribute to a shared peak charge, capturing the need to coordinate flexible electricity demand across overlapping workloads.  Such an extension might also consider how energy storage could help ``shave the peak.''  On the learning-augmented side, our algorithms focus on decision advice and use a parameter to control the consistency-robustness trade-off; it would be interesting to investigate other advice models in this setting, including uncertainty quantified predictions~\cite{Johnson:26UQ, Sun:24:UQ} and/or instance advice that might be easier to obtain for certain applications of \paowt.

%% file: acknowledgments.tex
This research is supported by National Science Foundation grants CAREER-2045641, CNS-2325956, CNS-2512128, and CNS-2533814, and the U.S. Department of Energy, Office of Science, Office of Advanced Scientific Computing Research, Department of Energy Computational Science Graduate Fellowship under Award Number DE-SC0024386. In addition, this work is funded by Natural Sciences and Engineering Research Council of Canada Discovery Grant RGPIN-2026-06424.

%% file: sec-ai-use.tex
The authors used generative AI tools in the creation of this paper and detail their role here. ChatGPT models produced \rob and the matching lower bound of \autoref{thm:known-T-lb}. These replaced an algorithm and lower bound which the authors had designed, which achieved the same competitive results but required a more involved analysis. ChatGPT also produced \decision and its matching consistency-robustness lower bound \autoref{thm:advice-gen-lb}, for which there was no earlier author-designed version. ChatGPT produced the lower bounds on instance advice found in  \autoref{lem:instanceOpCost} and \autoref{lem:instanceOpCount}. The experimental code was generated by Codex, and AI tools were used to assist the literature search. The authors formulated the problem and model, designed the experimental study, independently checked all proofs, and substantially rewrote the generated material. The authors take full responsibility for the presentation and correctness of all content.

%% file: appendix-preliminaries.tex
In this section, we give the full proofs of results deferred from \autoref{sec:prob}.

\unknownTLB*

\begin{proof}
    Fix a deterministic online algorithm, and consider a family of instances indexed by $k \in \N$ that begin by serving the price $p_t = 1/k$ for $t = 1, \ldots, k$. Let $x_t^{(k)}$, $t \in [k]$ denote the amount the algorithm purchases on instance $k$ in each of the first $k$ timesteps. We break into two cases.
    \begin{enumerate}
        \item Suppose there exists some $k' \in \N$ such that $\sum_{t=1}^{k'} x_t^{(k')} > \frac{1}{2}$. Then $\max_{t \in [k']} x_t^{(k')} \geq \frac{1}{2k'}$. Conclude the instance by presenting prices $p_t = 0$ for $t = k'+1, \ldots, T$, with some arbitrarily large $T$. Then the algorithm incurs cost at least
        \begin{align*}
            \frac{1}{k'}\sum_{t=1}^{k'} x_t^{(k')} + 0\cdot\left(1-\sum_{t=1}^{k'} x_t^{(k')}\right) + C \cdot \max_{t \in [k']} x_t^{(k')} &\geq \frac{1}{2k'}(1+C),
        \end{align*}
        while $\Opt$ has cost at most $0 + \frac{C}{T-k'}$. Thus, the competitive ratio is at least
        $$\frac{(1+C)(T-k')}{2k'C} = \Omega(T)$$
        as $T \to \infty$. 
        
        \item Suppose alternatively that there exists no $k' \in \N$ such that $\sum_{t=1}^{k'} x_t^{(k')} > \frac{1}{2}$. Then for any $k \in \N$, conclude the instance with a single price $p_{k+1}=U$. The algorithm incurs cost at least $U/2$, while buying $1/k$ at each of the $k$ earlier prices is feasible with cost $k(1/k)(1/k)+C/k=(1+C)/k$, so $\C{\Opt}\le(1+C)/k$. Setting $T=k+1$, the competitive ratio is at least
        $$\frac{U(T-1)}{2(1+C)} = \Omega(T)$$
        as $T \to \infty$ (i.e., as $k \to \infty$).
        
    \end{enumerate}
\end{proof}

\LargeCLB*
\begin{proof}
    Order the prices so that $p_1\le\cdots\le p_T$, allowing ties. We first show why it suffices to compare uniform purchasing with plans that buy equally in the cheapest rounds. Since $p_t\ge0$ and $C>0$, an optimal plan can be chosen to buy exactly one unit. Fix its largest purchase $v\in[1/T,1]$. Subject to buying at most $v$ in any round, it is cheapest to buy at the lowest prices first. For $1/k\le v\le1/(k-1)$, where $2\le k\le T$, this gives purchases of $v$ in the first $k-1$ rounds, $1-(k-1)v$ in round $k$, and zero thereafter. Their total cost is
    \[
        p_k+v\left(C+\sum_{i=1}^{k-1}p_i-(k-1)p_k\right),
    \]
    which is linear in $v$. Its minimum over that interval therefore occurs at $v=1/k$ or $v=1/(k-1)$. These intervals cover all possible values of $v$, so at least one optimal plan buys $1/k$ in each of the $k$ cheapest rounds for some $k\in[T]$, at cost $(\sum_{i=1}^k p_i+C)/k$.

    Now assume $C\ge(T-1)(U-L)$. For every $k\in[T-1]$, we see that purchasing across the $k$ cheapest rounds is at least as expensive as uniform purchasing across $T$:
    \begin{align*}
        C &\ge k(U-L)\\
          &\ge k\left(U-\frac{1}{k}\sum_{i=1}^k p_i\right),\\
        \frac{T-k}{Tk}C
          &\ge \frac{T-k}{T}\left(U-\frac{1}{k}\sum_{i=1}^k p_i\right),\\
        \left(\frac{1}{T}-\frac{1}{k}\right)C
          &\le \frac{T-k}{T}\left(\frac{1}{k}\sum_{i=1}^k p_i-U\right),\\
        \frac{1}{T}\left(\sum_{i=1}^k p_i+(T-k)U+C\right)
          &\le \frac{1}{k}\left(\sum_{i=1}^k p_i+C\right).
    \end{align*}
    Because $\sum_{i=k+1}^T p_i\le(T-k)U$, the cost of buying $1/T$ in every round satisfies
    \[
        \frac{1}{T}\left(\sum_{i=1}^T p_i+C\right)
        \le \frac{1}{T}\left(\sum_{i=1}^k p_i+(T-k)U+C\right)
        \le \frac{1}{k}\left(\sum_{i=1}^k p_i+C\right).
    \]
    This also holds with equality when $k=T$. Since some optimal plan has one of these forms, uniform purchasing is optimal. An online player that buys $1/T$ in each round is therefore 1-competitive in this regime.
\end{proof}

%% file: appendix-learning-aug-algs.tex
In this section, we give the full proofs of results deferred from \autoref{sec:decision}.

\decisionConsistent*

    \begin{proof}
        Let $A$ be the cost when directly following advice. We consider two cases.
        
        \paragraph{Case 1} If $A\le H$ then it is the case that every shadow variable $y_t=\hat{x}_t$ and that $\hat{x}_t\le m$ for all $t$ because the peak charge is at most $H = Cm$. Therefore, $x_t \ge \hat{x}_t$ up until the step where constraint is satisfied. The $\decision$ algorithm may buy more when $d_t > \hat{x}_t$, displacing at most one unit of lower prices (raising cost at most $(\Phi-L)$) and increasing the peak purchase by at most $\nicefrac{(\Phi-L)}{C}$.

        $$\C{\decision}\le A+2(\Phi-L)\le (1+\frac{2(\Phi-L)}{L+C/T})A$$

        \paragraph{Case 2} If $A > H$, we instead use the trivial bound

        $$\C{\decision}\le U+Cm \le U+H \le (1+\frac{U}{H})A$$
    \end{proof}

\decisionRobust*

    \begin{proof}
        As in \autoref{thm:rob-ub} we will split the analysis into cases of whether a compulsory trade was needed.

        \paragraph{No Compulsory Trade} By construction, we can see that the cost of the shadow decisions is at most $H$. In each time step, if \decision purchases more than the shadow decision ($x_t>y_t$) it must be the case that $x_t \le d_t$, so the price $p_t < \Phi$. The total mass of such purchases is at most 1, and the peak purchase of \decision is at most $\nicefrac{(\Phi-L)}{C}$ greater than the shadow algorithm. Thus, $\C{\decision} \le H+\Phi + (\Phi-L)$. Meanwhile we can trivially lower bound $\C{\Opt} \ge L+\nicefrac{C}{T}$.
    
        \paragraph{Compulsory Trade} Let $t$ be the first round in which the compulsory purchase exceeds both the price-threshold decision $d_t=(\Phi-p_t)_+/C$ and the shadow decision. Earlier purchases total $1-r$ and each is at least its threshold decision. Since $d_s\le(\Phi-L)_+/C\le m$ in future rounds, the compulsory condition $d_t<r-(T-t)m$ implies
        \[
            \sum_{s=1}^T d_s\le (1-r)+d_t+(T-t)m<1.
        \]
        To bound \Opt directly, let $z_s$ be any feasible schedule and set $v=\max_s z_s$. Since $(\Phi-p_s)z_s\le(\Phi-p_s)_+v$ for every $s$, we have
        \[
            \sum_{s=1}^T p_s z_s+Cv
            \ge \Phi\sum_{s=1}^T z_s+
                \Bigl(C-\sum_{s=1}^T(\Phi-p_s)_+\Bigr)v
            \ge \Phi.
        \]
        The last step uses $\sum_s z_s\ge1$, $\Phi\ge0$, and $\sum_s(\Phi-p_s)_+=C\sum_s d_s<C$. Hence $\C{\Opt}\ge\Phi$. No purchase exceeds $m$ when $m\ge1/T$, so $\C{\decision}\le U+Cm\le U+H$.
    
        Combining these results, we see that for a given $\Phi,H$ the \decision algorithm is $\gamma$-robust for:
    
        $$\gamma = \max\left\{\frac{H+2\Phi-L}{L+C/T}, \frac{U+H}{\Phi}\right\} $$    
    \end{proof}

\thmInstance*

\begin{proof}
    We begin by showing that having only one of  $(\hat{\lambda}, \hat{k})$ is insufficient to reach 1-consistency.  

    \begin{restatable}{lemma}{lemInstanceOptCost}\label{lem:instanceOpCost}
    Consider any learning-augmented algorithm with predicted optimal cost $\hat{\lambda} = \C{\Opt}$. The consistency of \Alg is $1+\epsilon \ge \frac{4}{3}-o(1)$.
    \end{restatable}

    \begin{proof}
    Consider the following families of sequences for $T\ge 3$ and $C=2$

    \[
    \mathcal{I}_1^{(T)} = (0,0,\underbrace{2,\dots,2}_{T-2}) \quad \quad
    \mathcal{I}_2^{(T)} = (0,\underbrace{1-\frac{1}{T-1},\dots,1-\frac{1}{T-1}}_{T-1})
    \]

    Both of these instance families have optimal cost $1$ regardless of $T$. For $\mathcal{I}_1^{(T)}$ set $x = (\nicefrac{1}{2},\nicefrac{1}{2},0,\dots,0)$, and for $\mathcal{I}_2^{(T)}$ set $x = (\nicefrac{1}{T},\dots,\nicefrac{1}{T})$. 
    
    Now, consider the purchases of learning-augmented $\Alg$. On both instances $p_1=0$, so $\Alg$ will purchase a set amount $x_1$. For $\mathcal{I}_1^{(T)}$, let $z$ denote the total amount that $\Alg$ purchases across price $2$. Thus, on $\mathcal{I}_1^{(T)}$, $\Alg$ will purchase $x_2 = 1-x_1-z$. We can lower bound the cost on $\mathcal{I}_1^{(T)}$ as follows:

    $$\C{\Alg(\mathcal{I}_1^{(T)})} \ge 2z+2x_2 = 2z+2(1-x_1-z)=2-2x_1$$

    Meanwhile, we can lower bound the cost on $\mathcal{I}_2^{(T)}$ as 
    $$\C{\Alg(\mathcal{I}_2^{(T)})} \ge (1-x_1)(1-\frac{1}{T-1})+2x_1 = 1+\frac{Tx_1-1}{T-1}$$

    Thus, the consistency can be lower bounded as $\max\{2-2x_1,1+\frac{Tx_1-1}{T-1}\}$, which is minimized by equalizing the expressions and solving for $x_1$. 
    
    \begin{align*}
        2-2x_1 &= 1+\frac{Tx_1-1}{T-1} \\
        (T-1)(1-2x_1)&=Tx_1-1 \\\
        T &= Tx+2Tx-2x\\
        x&= \frac{T}{3T-2}
    \end{align*}
    Substituting this value into either side of the $\max$ expression gives consistency at least $\frac{4T-4}{3T-2}$. Since $T$ can be made arbitrarily large, $1+\epsilon \ge \frac{4}{3}-o(1)$.

    \end{proof}

    \begin{restatable}{lemma}{lemInstancePurchaseCount}\label{lem:instanceOpCount}
     Consider any learning-augmented algorithm with predicted optimal purchase count $\;\hat{k} = |\{p_t \le \C{\Opt}\}|$. The consistency of \Alg is at least $1+\frac{1}{4\hat{k}-3}$.
    
    \end{restatable}

    \begin{proof}
        We begin by establishing a helpful representation for the difference between $\C{\Alg}$ and $\C{\opt}$, where $m$ is the max purchase of $\Alg$.

        \begin{equation}\label{eq:cost-decomp}
            \C{\Alg} - \lambda^* = \sum(p_t-\lambda^*)_+ +\sum(\lambda^*-p_t)_+(m-x_t)
        \end{equation}

        The equation holds by decomposing the price cost of \Alg as $\sum p_tx_t = \lambda^*+\sum(p_t-\lambda^*)_+x_t-\sum(\lambda^*-p_t)x_t$, and recalling that $\lambda^*$ is specifically the value where $\sum(\lambda^*-p_t)=C$.

        Now, let $L=0$ and set parameters $U,C$ such that $U \ge \frac{4C}{\hat{k}}$ for $2 \le k \le T-1$. In this region, note that both $p = \frac{2C}{k}$ and $q=\frac{4C}{k}$ are within $[L,U]$. Consider the following families of sequences indexed by $T$.

    \[
    \mathcal{I}_1^{(T)} = (\underbrace{p,\dots, p}_{k},\underbrace{q,\dots,q}_{T-k}) \quad \quad
    \mathcal{I}_2^{(T)} = (p,\underbrace{L,\dots,L}_{k},\underbrace{q,\dots,q}_{T-k-1})
    \]

    Note that the optimal solution to $\mathcal{I}_1^{(T)}$ is to purchase $\nicefrac{1}{k}$ at each of the prices $p$, for total cost $\nicefrac{3C}{k} < q$. Meanwhile, the optimal solution to $\mathcal{I}_2^{(T)}$ is to purchase $\nicefrac{1}{k}$ at each of the prices $0$, for total cost $\frac{C}{k}<p$.

    Upon observing $p_1 = p$, \Alg cannot differentiate between the possible instances, and will purchase a constant $x_1$. Recalling \ref{eq:cost-decomp}, we can state that the cost difference between \Alg and \Opt on $\mathcal{I}_2^{(T)}$ thus satisfies:

    $$ \C{\Alg(\mathcal{I}_2^{(T)})}-\C{\Opt(\mathcal{I}_2^{(T)})} \ge \sum(p_t-\frac{C}{k})_+x_t \ge (p_1-\frac{C}{k})_+x_1$$

    Since $p_1 = \frac{2C}{k}>\frac{C}{k}=\lambda^*$, we can thus claim $\C{\Alg(\mathcal{I}_2^{(T)})}-\C{\Opt(\mathcal{I}_2^{(T)})}\ge \frac{C}{k}x_1$ and $\frac{\C{\Alg(\mathcal{I}_2^{(T)})}}{\C{\Opt(\mathcal{I}_2^{(T)})}} \ge 1+x_1$.

    Now considering $\mathcal{I}_1^{(T)}$, let $y$ be the total volume purchased during prices $q$, and recall $q>\lambda^*$. Thus, $1-x-y$ was purchased during remaining prices $p$, meaning max purchase $m\ge \frac{1-x_1-y}{k-1}$. Recalling \ref{eq:cost-decomp}, we can state that the cost difference between \Alg and \Opt on $\mathcal{I}_1^{(T)}$ is 

    \begin{align*}
        \C{\Alg(\mathcal{I}_1^{(T)})}-\C{\Opt(\mathcal{I}_1^{(T)})} &= \sum(p_t-\frac{3C}{k})_+x_t+{}\sum(\frac{3C}{k}-p_t)_+(\frac{1-x_1-y}{k-1}-x_t) \\
        &\ge y\frac{C}{k}+\frac{C}{k}\sum_{t:p_t=p}(\frac{1-x_1-y}{k-1}-x_t)\\
        \text{Since $p$ appears }& k\text{ times} \\
         &\ge \frac{C}{k}[y+k\frac{1-x_1-y}{k-1}-(1-y)]\\
        &\ge \frac{C}{k}[\frac{1-kx_1+(k-2)y}{k-1}]\\
        \text{Since } &k\ge 2 \\
        &\ge \frac{C}{k(k-1)}(1-kx)
    \end{align*}

    Therefore, we can bound the competitive ratio on $\mathcal{I}_1^{(T)}$ as $\frac{\C{\Alg(\mathcal{I}_1^{(T)})}}{\C{\Opt(\mathcal{I}_1^{(T)})}} \ge 1+\frac{1-kx_1}{3(k-1)}$.

    Combining these results, we find consistency is lower bounded by $ 1+\max\{x_1, \frac{1-kx_1}{3(k-1)}\}$. This value is clearly minimized when $x_1= \frac{1-kx_1}{3(k-1)}$. Substituting gives the result that consistency with only purchase count advice is at least $1+\frac{1}{4k-3}$.
    \end{proof}

Now that we have established that either prediction is insufficient by itself, we can observe that combining optimal cost and purchase count predictions results in a 1-competitive algorithm. Given perfect predictions $(\hat{\lambda} = \lambda^*, \hat{k} =k)$, consider an algorithm which purchases $x_t=\nicefrac{1}{\hat{k}} \mathbbm{1} \{p_t \le \hat{\lambda}\}$. We can see this is exactly the structure of the optimal primal solution given in \autoref{eq:offline-opt-solution}. Thus, $\C{\Alg} = \C{\Opt}$ when both predictions are perfectly accurate, giving 1-consistency.
\end{proof}